\documentclass[conference]{IEEEtran}
\IEEEoverridecommandlockouts
\usepackage[utf8]{inputenc}
\usepackage[T1]{fontenc}
\usepackage{cite}
\usepackage{amsmath,amssymb,amsfonts,amsthm,bm,bbm,mathtools}
\usepackage{graphicx}
\usepackage{textcomp}
\usepackage{xcolor}
\usepackage{hyperref}
\usepackage{subfig}
\usepackage{balance}
\usepackage[dvipsnames]{xcolor}

\newtheorem{theorem}{Theorem}
\newtheorem{lemma}{Lemma}

\newtheorem{prop}{Proposition}
\usepackage{algorithm,algorithmic}
\usepackage{color,soul}
\usepackage{graphicx}
\usepackage{subfig}
\usepackage{hyperref}
\newtheorem{Ass}{Assumption}

\newtheorem{remark}{Remark}

\usepackage{color,soul}

\newcommand{\R}{\mathbb{R}}
\newcommand{\1}{\bm{1}}
\newcommand{\wt}{\mathbf{w}}
\newcommand{\wbar}{\overline{w}}
\newcommand{\Fbar}{\overline{F}}
\newcommand{\fbar}{\overline{f}}

\newcommand{\norm}[1]{\left\lVert #1 \right\rVert}

\newcommand{\mat}[1]{\mathbf{#1}}

\title{Decentralized Time-Varying Optimization for Streaming Data via Temporal Weighting}

\author{%
  Muhammad Faraz Ul Abrar$^{*}$, Nicol\`o Michelusi$^{*}$, and Erik G. Larsson$^{\dagger}$\\
  $^{*}$School of Electrical, Computer and Energy Engineering, Arizona State University\\
  $^{\dagger}$Department of Electrical Engineering (ISY), Link\"oping University\\
  Emails: mulabrar@asu.edu, nicolo.michelusi@asu.edu, erik.g.larsson@liu.se

\thanks{This research has been funded in part by NSF under grant CNS-2129015. The work of E. G. Larsson was supported in part by ELLIIT, VR, and the KAW foundation. }
}

\begin{document}
\maketitle

\begin{abstract}
Classical optimization theory largely focuses on fixed objective functions, whereas many modern learning systems operate in dynamic environments where data arrive sequentially and decisions must be updated continuously. In this work, we study optimization with streaming data over a distributed network of agents. We adopt a structured, \emph{weight-based} formulation that explicitly captures the streaming-data origin of the time-varying objective: at each time step, every agent receives a new sample, and the network seeks to track the minimizer of a \emph{temporally weighted} objective formed from all samples observed across the network so far. We focus on decentralized gradient descent (DGD) with a limited communication/computation budget, where at each time step, only a limited number of DGD iterations can be performed before the objective changes again. For strongly convex and smooth losses, we analyze the \emph{tracking error} with respect to the time-varying minimizer through a fixed-point theory lens. Our analysis reveals that the tracking error decomposes into a fixed-point tracking term and a bias term induced by data heterogeneity across agents. We specialize the analysis to two natural weighting strategies: uniform weights, which treat all samples equally, and exponentially discounted weights, which geometrically decay the influence of older data. Under uniform weighting, DGD tracks the fixed-point at a rate $\mathcal{O}(1/t)$, whereas discounted weighting yields a non-vanishing fixed-point tracking floor controlled by the discount factor. 
 In both cases, decentralization induces an additional non-zero bias floor under a constant step size.
We validate our theoretical findings through numerical simulations. 
\end{abstract}

\begin{IEEEkeywords}
Decentralized optimization, streaming data, time-varying optimization, DGD, tracking error, online learning.
\end{IEEEkeywords}
\vspace{-5mm}

\section{Introduction}
\label{sec:Intro}
Modern learning systems increasingly rely on optimization to support decision-making in applications such as robotics, telecommunications, cyber-physical systems, and autonomous control \cite{Survey_dist_optm,Time_Structured}. While conventional formulations typically assume a static objective and a fixed data distribution, most real-world systems operate in dynamic environments: data arrive sequentially, the learning objective evolves over time, and decisions must be continually updated rather than computed once and deployed indefinitely \cite{DallAnese2019_stream,OL_survey,CL_Survey}. At the same time, privacy, communication, or architectural constraints often require data acquisition and processing to be carried out across multiple agents in a decentralized manner \cite{Yuan_DFL_Survey2024,FL_Survey,DL_wireless}. These features give rise to distributed learning systems that must adapt online to streaming information while cooperating over a communication network. 

A standard way to model such problems is through time-varying optimization \cite{Time_Structured}, in which the goal is to track the minimizer of a \emph{time-varying} global objective $F_t$ of the form:
\begin{align}
\overline{\mathbf w}_t^* \in \arg\min_{\mathbf w\in\R^d} F_t(\mathbf w)
\triangleq \frac{1}{N}\sum_{n=1}^N \widetilde f_{n,t}(\mathbf w), \qquad t\ge 1,
\label{eq:intro_generic_tvo}
\end{align}
where $\widetilde f_{n,t}(\cdot)$ denotes the local objective at agent $n$ at time $t$. 
This problem has been studied extensively in both single-agent and multi-agent settings; see, e.g., \cite{Polyak_book,Popkov,Time_Structured,Class_Prediction_Correction,QLing2014ADMM,Usman2016DistributedDirected}.
The typical procedure to solve \eqref{eq:intro_generic_tvo} is to apply an iterative optimization method such as (decentralized) gradient descent \cite{Nedic2009DGD}. However, because the objective changes over time, real-time implementations typically permit only a limited number of iterations before the problem evolves again \cite{DallAnese2019_stream,Time_Structured}. Exact tracking of the time-varying minimizer is therefore generally infeasible, and performance is naturally quantified through the \emph{tracking error} (TE), i.e., the distance between the current iterate and the current optimizer $\overline{\mathbf w}_t^*$.

Early works on time-varying optimization focused on the single-agent setting \cite{Polyak_book,tracking_minimum_1998,Popkov}. For smooth and strongly-convex objectives, it is known that gradient descent with a constant step size tracks the time-varying minimizer $\overline{\mathbf w}_t^*$ up to an $\mathcal{O}(C)$ neighborhood when the minimizer drift is uniformly bounded, i.e., for all $t \geq1$, $\|\overline{\mathbf w}_{t+1}^*-\overline{\mathbf w}_t^*\|\le C$ \cite{Polyak_book,Class_Prediction_Correction,Time_Structured}. Analogous conclusions hold in decentralized settings with an additional non-vanishing error caused by network disagreement; see, e.g., \cite{Usman2016DistributedDirected,QLing2020}. 
Despite the diversity of algorithms considered, all these works adopt the generic model \eqref{eq:intro_generic_tvo}, thereby treating the evolution of $\{F_t\}_t$ as an arbitrary time-varying process. While broadly applicable, this abstraction can be conservative in learning settings where the temporal variation is induced by a specific mechanism, such as the sequential arrival of data. Consequently, in such settings, exploiting these structured objective evolutions can lead to sharper and more interpretable TE guarantees.

Motivated by this observation, and building on our prior single-agent work \cite{TV_asilomar25}, we study decentralized optimization over a multi-agent network from a streaming-data perspective. Specifically, we consider a \emph{structured} time-varying formulation in which the objective at time $t$ is a \emph{temporally weighted} average of the instantaneous network losses observed up to time $t$, with the weights encoding the relative temporal importance of past data. Under this formulation, we analyze multi-iteration DGD under a limited communication/computation budget and derive \emph{weight-specific} tracking guarantees. We then specialize our analysis to two canonical weighting strategies: uniform weights, which assign equal importance to all past samples and model stationary environments, and exponentially discounted weights, which geometrically decay past sample contributions, prioritizing recent observations. Using the fixed-point operator viewpoint of \cite{MHT}, we show that uniform weighting yields a vanishing fixed-point tracking contribution with rate $\mathcal{O}(1/t)$, whereas discounted weighting induces a non-vanishing steady-state tracking floor controlled by the discount factor. In both cases, an additional bias term is present and is controlled by the learning step size.



\paragraph*{Notation} Scalars are denoted by italic letters, vectors by boldface lowercase letters, and matrices by boldface uppercase letters. The all-ones vector is denoted by $\1$, and $\mat{I}$ denotes the identity matrix. For a vector $\mathbf w\in\R^N$, its $n$-th component is denoted by $[\mathbf w]_n$. The Euclidean norm is denoted by $\norm{\cdot}$, transpose is denoted by $(\cdot)^\top$ and for a matrix $\mat{A}$, $\norm{\mat{A}}_2$ denotes its spectral norm.
\vspace{-1.5mm}
\section{System Model}
\label{sec:model}
We consider a network of $N$ agents connected over an undirected, connected, time-invariant graph $\mathcal G=(\mathcal V,\mathcal E)$, where $\mathcal V=\{1,\dots,N\}$ is the set of agents and
$\mathcal{E}\subseteq\mathcal{V}\times\mathcal{V}$ is the edge set. In particular, agents $m$ and $n$ can exchange information if and only if $(m,n)\in\mathcal{E}$. 
The communication across agents is captured by a symmetric (not necessarily entrywise nonnegative) 
mixing matrix $\mat{M}\in\R^{N\times N}$ induced by the network topology, satisfying
\begin{align}
\mat{M}=\mat{M}^\top, \qquad \mat{M}\1=\1.
\label{eq:mixing_matrix_assump}
\end{align}
We further assume that the eigenvalues of $\mat{M}$ lie in $(-1,1]$ and are ordered as  $1=\lambda_1(\mat{M})>\lambda_2(\mat{M})\ge \cdots \ge \lambda_N(\mat{M})>-1$.

We next introduce the streaming-data model that induces the time variation of the objective. While the generic problem in \eqref{eq:intro_generic_tvo} is posed over $\R^d$, in this paper we focus on the scalar case $w\in\R$.\footnote{The extension to vector-valued decision variables is conceptually identical but requires heavier notation.} At each time $t\ge 1$, every agent $n\in\{1,\dots,N\}$ observes a new data sample and incurs an instantaneous loss $\ell_{n,t}(w)$. We define the network-wide instantaneous loss as
\begin{align}
F_t(w) \triangleq \frac{1}{N}\sum_{n=1}^N \ell_{n,t}(w).
\label{eq:instantaneous_network_loss}
\end{align}
The goal of the agents is to collaboratively obtain the minimizer of the global objective $\Fbar(\cdot)$ defined as a temporally weighted average of the past instantaneous losses. This leads to the following optimization problem:
\begin{align}
\wbar_t^* \in \arg\min_{w\in\R} \Fbar_t(w), \quad\Fbar_t(w) \triangleq \sum_{i=1}^t a_i(t) F_i(w), 
\label{eq:main_tv_prob}
\end{align}
where $a_i(t)$ denotes the weight assigned at time $t$ to the network loss associated with the samples observed at time $i$. The temporal weights are assumed to satisfy 
\begin{align}
a_i(t)\in[0,1], \qquad \sum_{i=1}^t a_i(t)=1,
\label{eq:weight_simplex}
\end{align}
for all $t\ge 1$. Note that since the global objective $\Fbar_t(\cdot)$ in \eqref{eq:main_tv_prob} evolves with time, so does its minimizer $\wbar_t^*$. Our objective is to understand how accurately a decentralized optimization method can \emph{track} $\wbar_t^*$ under a limited communication and computation budget per time step.






\section{Decentralized Gradient Descent for Streaming Data}
\label{sec:alg_perf}
We now describe a DGD-type update for the streaming time-varying problem in \eqref{eq:main_tv_prob}, based on the classical DGD method \cite{Nedic2009DGD}.
To this end, each agent maintains a local parameter $w_{n,t}$ at time $t$. For the analysis, define the vector of stacked local parameters $\wt_t \triangleq [w_{1,t},\dots,w_{N,t}]^\top \in \R^N$. We also define the component-wise instantaneous objective
$f_t(\mathbf w) \triangleq \sum_{n=1}^N \ell_{n,t}(w_n)$
and its temporally weighted counterpart
\begin{align}
\fbar_t(\mathbf w)
\triangleq \sum_{i=1}^t a_i(t) f_i(\mathbf w)
=\sum_{i=1}^t a_i(t)\sum_{n=1}^N \ell_{n,i}(w_n).
\label{eq:stacked_weighted_obj}
\end{align}
Note that the objectives $\Fbar_t(\cdot)$ and $\fbar_t(\cdot)$ are built from the same collection of sample losses $\{\ell_{n,i}\}$, but they act on different variables. In particular, $\Fbar_t$ is defined over a common scalar decision variable $w\in\R$, whereas $\fbar_t$ is defined over the stacked vector of local parameters $\mathbf w=(w_1,\dots,w_N)^\top\in\R^N$.
On the consensus subspace
$\{\mathbf w\in\R^N : \mathbf w = w\1\}$, however, the two objectives are related as $\fbar_t(w\1)=\sum_{i=1}^t a_i(t)\sum_{n=1}^N \ell_{n,i}(w)=N\Fbar_t(w)$. Hence, when all agents agree on a common parameter, $\fbar_t$ and $\Fbar_t$ share the same minimizers. Finally, we denote the stacked time-$t$ global minimizer of \eqref{eq:main_tv_prob} by $\overline{\mathbf w}_t^* \triangleq \overline w_t^*\,\bm{1}\in\mathbb{R}^N$. 

After the new samples arrive at time $t+1$, the network performs $E$ DGD iterations on the updated temporally weighted objective. Starting from ${\wt_t^{(0)}{=}\wt_t}$, the inner iterates evolve as
\begin{align}
\wt_t^{(k+1)} = \mat{M}\wt_t^{(k)} - \eta \nabla \fbar_{t+1}(\wt_t^{(k)}),
\label{eq:dgd_inner_updates}
\end{align}
for all $k=0,\dots,E-1$, where $\eta>0$ is a constant step size. In particular, at each inner iteration, the agents combine their current parameters with those of their neighbors through the mixing matrix $\mat{M}$ and then take a gradient step on the updated temporally weighted objective. After $E$ such inner iterations, the local parameters are updated as $\wt_{t+1}=\wt_t^{(E)}$. Notably, the update in \eqref{eq:dgd_inner_updates} requires evaluating the gradient of the temporally weighted objective, which in principle depends on all samples observed up to time $t+1$. Since our focus is on the fundamental tracking limits of decentralized learning from streaming data, we do not impose memory constraints in the present analysis. 
Memory-constrained extensions are outside the scope of this paper and will be addressed in ongoing work 
\cite{Decentralized_TV_journal_prep}.

Owing to the time-varying nature of \eqref{eq:main_tv_prob} and constrained communication/computation budget per time step, we characterize the performance using the tracking error, defined as
\begin{align}
\mathrm{TE}(t) \triangleq \norm{\wt_t - \overline{\mathbf w}_t^*},
\label{eq:TE_def}
\end{align}
and the asymptotic tracking error, defined as $\mathrm{ATE} \triangleq \limsup_{t\to\infty} \norm{\wt_t-\overline{\mathbf w}_t^*}$, quantifying how well the decentralized iterates track the time-varying optimizer induced by the streaming-data process. 

\section{Tracking Error Analysis}
\label{sec:analysis}
In this section, we analyze the DGD updates \eqref{eq:dgd_inner_updates} for solving the streaming time-varying problem \eqref{eq:main_tv_prob}. We begin by deriving generic TE bounds and then specialize them to two natural temporal weighting strategies. For the analysis, we require the following standard assumptions, widely adopted in analyzing gradient-based methods; see, e.g., \cite{Nesterov,boyd2004convex,MHT}.

\begin{Ass}
\label{ass:smooth_sc}
For every agent $n$ and time $t$, the loss $\ell_{n,t}(\cdot)$ is $L$-smooth and $\mu$-strongly convex, i.e., for all $x,y\in\R$, 
\begin{align} \left| \ell_{n,t}'(x) - \ell'_{n,t}(y)\right| \leq L| {x} - {y}|,\end{align}
\begin{align}
\ell_{n,t}({y}) \geq \ell_{n,t}({x}) +  \ell_{n,t}'(x)({y} - {x}) + \frac{\mu}{2}| y- {x} |^2.
\end{align}  
\end{Ass}
Throughout the analysis, we use $\kappa \triangleq L/\mu$ to denote the condition number of the optimization problem.
Under Assumption~\ref{ass:smooth_sc}, both $\Fbar_t(\cdot)$ and $\fbar_t(\cdot)$ are also $L$-smooth and $\mu$-strongly convex for every $t\geq1$. Likewise, it is easy to verify that the component-wise objectives $f_t(\cdot)$ and $\fbar_t(\cdot)$ are also $L$-smooth and $\mu$-strongly convex for all $t\geq 1$. That is, for all $\mathbf u,\mathbf v\in\mathbb{R}^N$ and for either choice
$\tilde f_t \in \{f_t,\overline{f}_t\}$,
\begin{align}
\bigl\|\nabla \tilde f_t(\mathbf{u}) - \nabla \tilde f_t(\mathbf{v})\bigr\|
   \leq L \|\mathbf{u}-\mathbf{v}\|,
\label{eq:tilde_f_t_smooth}
\end{align}
\begin{align}
\tilde f_t(\mathbf{v})
   \ge \tilde f_t(\mathbf{u})
      +  \nabla \tilde f_t(\mathbf{u})^\top (\mathbf{v}-\mathbf{u})
      + \frac{\mu}{2} \|\mathbf{v}-\mathbf{u}\|^2.
\label{eq:tilde_f_t_sc}
\end{align}
To analyze the TE in \eqref{eq:TE_def}, we adopt a fixed-point operator theory viewpoint \cite{MHT} to study the DGD updates \eqref{eq:dgd_inner_updates}. To this end, define the time-varying DGD mapping
\begin{align}
\phi_t(\mathbf w) \triangleq \mat{M}\mathbf w - \eta\nabla\fbar_t(\mathbf w).
\label{eq:dgd_map}
\end{align}
With this, the $E$ DGD updates in \eqref{eq:dgd_inner_updates} at time $t$ can be expressed as applying the operator $\phi_{t+1}$ to $\mathbf{w}_t$ $E$ times, i.e., \begin{align}
    \mathbf{w}_{t+1}
    = (\underbrace{\phi_{t+1}\circ\cdots\circ\phi_{t+1}}_{E\ \text{times}})
      (\mathbf{w}_t)
    \;\triangleq\; \Phi_{t+1}(\mathbf{w}_t), 
    \label{eq:Phi-def}
\end{align}
where $\Phi_{t}$ denotes the corresponding
composition mapping.

Importantly, under a suitable step-size condition, the DGD update operator $\phi_t$ is contractive, as formalized next.
\begin{lemma}
\label{lem:contraction}
Suppose Assumption~\ref{ass:smooth_sc} holds. Fix any $t\ge 1$ and consider $\phi_t$ in \eqref{eq:dgd_map}. If the step size is chosen such that $0<\eta \le \frac{1+\lambda_N(\mat{M})}{L+\mu}$, then $\phi_t$ is a contraction, i.e.,
for any $\mathbf u,\mathbf v\in\mathbb{R}^N$,
\begin{align*}
\|\phi_t(\mathbf u)-\phi_t(\mathbf v)\|\le (1-\eta\mu) \|\mathbf u-\mathbf v\|,
\end{align*}
with $1-\eta\mu \in (0,1)$.
Consequently, $\Phi_t$ defined in \eqref{eq:Phi-def} is also a contraction, with the contraction factor $\alpha \triangleq (1-\eta\mu)^E$. 
\end{lemma}
\noindent The proof of Lemma~\ref{lem:contraction} can be found in \cite{MHT}. Notably, under the conditions of Lemma~\ref{lem:contraction}, Banach's fixed-point theorem \cite{Rudin1976} guarantees that, for each $t$, there exists a unique fixed point
\begin{align}
    \widetilde{\mathbf{w}}_t
    \in \mathbb{R}^N
    \quad\text{s.t.}\quad
    \widetilde{\mathbf{w}}_t
    = \phi_t(\widetilde{\mathbf{w}}_t).
    \label{eq:fixed_point_def}
\end{align}
Moreover, since $\widetilde{\mathbf{w}}_t=\phi_t(\widetilde{\mathbf{w}}_t)$, repeated application implies $\widetilde{\mathbf{w}}_t=\Phi_t(\widetilde{\mathbf{w}}_t)$; hence, the same point is also the fixed point of the $E$-step operator $\Phi_t$. Thus, the point $\widetilde{\mathbf{w}}_t$ is the network iterate that DGD updates would converge to if the objective $\overline{F}_t$ were frozen at time $t$.

Next, building upon this fixed-point perspective, adding and subtracting the moving fixed point in \eqref{eq:TE_def}, along with applying the triangle inequality, yields the following TE decomposition: 
\begin{align}
    \mathrm{TE}(t) 
    &\leq
    \bigl\|\mathbf{w}_t - \widetilde{\mathbf{w}}_t \bigr\|
    +
    \bigl\|\widetilde{\mathbf{w}}_t - \overline{\mathbf{w}}_t^*\bigr\|.
    \label{eq:TE-decomposition}
\end{align}
The first term in \eqref{eq:TE-decomposition} measures how well the iterates track the moving fixed point when only $E$ inner DGD steps are performed per time-index, which we refer to as the \emph{fixed-point tracking error} $\mathrm{FPTE}(t) {\triangleq} \|\mathbf{w}_t - \widetilde{\mathbf{w}}_t\|$. The second term captures the mismatch between the fixed point and the stacked optimizer at time $t$, and thus represents the bias induced by decentralization of data across agents. To bound these two terms, we impose the following additional assumption.
\begin{Ass}
\label{ass:bounded_minimizers}
Let $
w_{n,t}^* \in \arg\min_{w\in\R}\ell_{n,t}(w)$ be the minimizer for the loss of agent $n$ at time $t$.
There exists a constant $C>0$ such that 
$ |w_{n,t}^*|\le C,\forall n,t$.
\end{Ass}
Assumption~\ref{ass:bounded_minimizers} is a structural condition on the sample losses, rather than a direct assumption on the drift of the minimizer sequence, which is more commonly adopted in the time-varying optimization literature; see, e.g., \cite{Time_Structured}. Similar bounded-minimizer assumptions have also been used in related analyses of streaming and federated learning \cite{TV_asilomar25,ChungHu25StreamFL}. In our setting, this condition ensures that both the time-varying optimizer and the fixed-point sequence remain uniformly bounded, which in turn yields the uniform gradient bounds stated next.
\begin{lemma}
\label{lem:boundedness_compact}
Under Assumptions~\ref{ass:smooth_sc} and~\ref{ass:bounded_minimizers}, for all $t\ge 1$,
\begin{align}
|\overline w_t^*|
&\le C\sqrt{\kappa},
\ \ \ 
\|\widetilde{\mathbf w}_t\|
\le C\sqrt{N\kappa}.
\label{eq:minimizer_fixedpoint-bnd}
\end{align}
Moreover, for all $i,t\geq 1$,
\begin{align}
\|\nabla\overline f_t(\widetilde{\mathbf w}_i)\|\le G,
\|\nabla f_t(\widetilde{\mathbf w}_i)\|\le G,\|\nabla \overline f_t(\overline{\mathbf w}_t^*)\|\le G,
\label{eq:grad-bnd-G}
\end{align}
where $G \triangleq2L\,C\sqrt{N\kappa}$.
\end{lemma}
\noindent
The proof is given in Appendix~\ref{app:A}. 
The bounds on the two terms in \eqref{eq:TE-decomposition} are developed next.
\subsection{Bias between the fixed point and the time-varying optimizer}
\label{subsec:bias_term}

We first provide a bound on the second term in \eqref{eq:TE-decomposition}, namely
$\|\widetilde{\mathbf w}_t-\overline{\mathbf w}_t^*\|$.
To this end, we invoke the fixed-point-to-optimizer bound developed in \cite{MHT}, specialized to the present setting.
\begin{lemma}
\label{lem:generic_bias_bound}
Fix any $t\ge 1$ and consider the objective $\fbar_t(\cdot)$ together with the fixed point $\widetilde{\mathbf w}_t$ defined in \eqref{eq:fixed_point_def}. Under the conditions of Lemma~\ref{lem:contraction},
\begin{align}
\|\widetilde{\mathbf w}_t-\overline{\mathbf w}_t^*\|
\le \eta \kappa\Lambda
\|\nabla \fbar_t(\overline{\mathbf w}_t^*)\|\le\eta \kappa\Lambda G ,
\label{eq:bias_bound}
\end{align}
where $\Lambda$ represents the network topology factor and  is given by 
$\Lambda=\frac{1}{1-\lambda_2(\mat{M})}$.
\end{lemma}
\noindent
The first inequality in \eqref{eq:bias_bound} follows directly from \cite{MHT} by applying the fixed-point bias bound to $\fbar_t(\cdot)$ at each fixed time $t$, whereas the second inequality uses \eqref{eq:grad-bnd-G}.

\begin{remark}
\label{rem:bias_interpretation}
The bound in \eqref{eq:bias_bound} shows that the bias scales linearly with the step size $\eta$, and depends on the condition number $\kappa$, the network factor $\Lambda$, and the gradient magnitude $\|\nabla \fbar_t(\overline{\mathbf w}_t^*)\|$, which reflects data heterogeneity across agents. In particular, if all agents observe identical losses, then $\nabla \fbar_t(\overline{\mathbf w}_t^*)=\mathbf 0$, and the bias vanishes.
\end{remark}

\subsection{Fixed-point tracking error}
Having bounded the bias term, we now derive a bound on the first term in \eqref{eq:TE-decomposition}.
 From \eqref{eq:Phi-def}, we have $\mathbf w_{t+1}=\Phi_{t+1}(\mathbf w_t)$. Next, since $\widetilde{\mathbf w}_{t+1}$ is a fixed point of $\phi_{t+1}$, it is also a fixed point of $\Phi_{t+1}$, i.e., $
\widetilde{\mathbf w}_{t+1}=\Phi_{t+1}(\widetilde{\mathbf w}_{t+1})$. Using these facts, and recalling from Lemma~\ref{lem:contraction} that $\Phi_{t+1}$ is a contraction with factor $\alpha=(1-\eta\mu)^E$, the FPTE at time $t+1$ satisfies
\begin{align}
\mathrm{FPTE}(t+1)
&= \|\mathbf w_{t+1}-\widetilde{\mathbf w}_{t+1}\| \nonumber\\
&= \|\Phi_{t+1}(\mathbf w_t)-\Phi_{t+1}(\widetilde{\mathbf w}_{t+1})\| \nonumber\\
&\overset{(a)}{\leq} \alpha \|\mathbf w_t-\widetilde{\mathbf w}_{t+1}\| \nonumber\\
&\overset{(b)}{\leq} \alpha \|\mathbf w_t-\widetilde{\mathbf w}_t\|
   + \alpha \|\widetilde{\mathbf w}_{t+1}-\widetilde{\mathbf w}_t\|,
\label{eq:fpte_recursion}
\end{align}
where $(a)$ invokes the contractive property of $\Phi_{t+1}$, and $(b)$ follows from the triangle inequality. Therefore,
\begin{align}
\mathrm{FPTE}(t+1)
\le \alpha\,\mathrm{FPTE}(t)
+\alpha\|\widetilde{\mathbf w}_{t+1}-\widetilde{\mathbf w}_t\|.
\label{eq:fpte_recursion_compact}
\end{align}
Using recursion, the FPTE in \eqref{eq:fpte_recursion_compact} can also be expressed as 
\begin{align}
\mathrm{FPTE}(t)
\le
\alpha^{t}\|\mathbf w_0-\widetilde{\mathbf w}_1\|
+\sum_{i=1}^{t-1}\alpha^{t-i}
\|\widetilde{\mathbf w}_{i+1}-\widetilde{\mathbf w}_i\|,
\label{eq:fpte_unrolled}
\end{align}
for all $t\ge 1$. We utilized that $\alpha^{t-1}\mathrm{FPTE}(1)\leq \alpha^{t}\|\mathbf w_0-\widetilde{\mathbf w}_1\|$ (cf. \eqref{eq:fpte_recursion}--(a)). Equation \eqref{eq:fpte_unrolled} shows that the bound on FPTE consists of two components: a geometrically decaying initialization term and an accumulated drift term governed by the moving fixed-point sequence $\{\widetilde{\mathbf w}_t\}_t$. We next derive a bound on the fixed-point drift.
\begin{lemma}
\label{lem:generic-drift}
Suppose Assumption~\ref{ass:smooth_sc} holds. Then, for all $t\ge 1$, the fixed points $\{\widetilde{\mathbf w}_t\}$ of the mappings $\{\phi_t\}$ satisfy
\begin{align}
\|\widetilde{\mathbf w}_{t+1}-\widetilde{\mathbf w}_t\|
\le \frac{1}{\mu}\,
\bigl\|\nabla \fbar_{t+1}(\widetilde{\mathbf w}_{t+1})
-\nabla \fbar_t(\widetilde{\mathbf w}_{t+1})\bigr\|.
\label{eq:generic_fp_drift}
\end{align}
\end{lemma}
\noindent
The proof is given in Appendix~\ref{app:A}.

\begin{remark}
Lemma~\ref{lem:generic-drift} shows that the drift of the fixed-point sequence is controlled by the temporal variation of the gradient of the time-weighted objective. In particular, if $\nabla \fbar_t(\cdot)$ changes slowly with $t$, then the corresponding fixed points $\{\widetilde{\mathbf w}_t\}$ cannot move too fast.
\end{remark}





\subsection{Specialization to canonical temporal weighting strategies}
We now specialize the generic fixed-point drift bound in \eqref{eq:generic_fp_drift} to two specific temporal weighting strategies. Combined with the recursion in \eqref{eq:fpte_unrolled}, this yields explicit \emph{weight-specific} TE guarantees.

\subsubsection{Uniform weights}
A natural choice in stationary environments is to assign equal importance to all samples observed so far, leading to
\begin{align}
a_i(t)=\frac{1}{t}, \qquad i=1,\dots,t.
\label{eq:uniform_weights}
\end{align}
Under \eqref{eq:uniform_weights}, the temporally weighted objective satisfies
\begin{align}
\fbar_{t+1}(\mathbf w)
&= \sum_{i=1}^{t+1}\frac{f_i(\mathbf w) }{t+1}
= \frac{1}{t+1}\sum_{i=1}^{t}f_i(\mathbf w)+\frac{1}{t+1}f_{t+1}(\mathbf w)\nonumber\\
&= \frac{t}{t+1}\fbar_t(\mathbf w)+\frac{1}{t+1}f_{t+1}(\mathbf w).
\label{eq:fbar_uniform_recursion}
\end{align}
Continuing from \eqref{eq:generic_fp_drift} and using \eqref{eq:fbar_uniform_recursion}, we can bound the fixed-point drift as
\begin{align}
    \|\widetilde{\mathbf w}_{t+1}-\widetilde{\mathbf w}_t\|  &\leq \frac{1}{\mu} \Big\Vert \frac{1}{t+1}
\nabla f_{t+1}(\widetilde{\mathbf w}_{t+1})
-\frac{1}{t+1}\nabla \fbar_t(\widetilde{\mathbf w}_{t+1})
\Big\Vert \nonumber \\
&\leq\frac{2G}{\mu(t+1)},\label{eq:uniform_fp_drift}
\end{align}
where the second inequality follows from the triangle inequality and the gradient bounds in Lemma~\ref{lem:boundedness_compact}. Substituting \eqref{eq:uniform_fp_drift} into \eqref{eq:fpte_unrolled}, the FPTE under uniform weights satisfies
\begin{align}
\mathrm{FPTE}(t)
\le \alpha^t\|\mathbf w_0-\widetilde{\mathbf w}_1\|
+\frac{2G}{\mu}\sum_{i=1}^{t-1}\frac{\alpha^{t-i}}{i+1}.
\label{eq:FPTE_uniform_pre}
\end{align}

\begin{prop}
\label{prop:S_uniform_bound}
For $\alpha\in(0,1)$, define $S(t)\;\triangleq\;\sum_{i=1}^{t-1}\frac{\alpha^{t-i}}{\,i+1\,}$. Then, for all $t\ge t_0\triangleq \bigl\lceil \frac{2\alpha}{1-\alpha}\bigr\rceil,$ it holds that 
\begin{align}
S(t)\le \frac{A}{t},
\label{eq:S_uniform_bound}
\end{align}
where $A\triangleq \max\left\{t_0S(t_0),\frac{2\alpha}{1-\alpha}\right\}$.
\end{prop}
\noindent
Proposition~\ref{prop:S_uniform_bound} is a direct restatement of the corresponding summation bound from the single-agent analysis in \cite[Prop.~1]{TV_asilomar25}, included here for completeness; its proof is therefore omitted.
Combining \eqref{eq:FPTE_uniform_pre} with Proposition~\ref{prop:S_uniform_bound} yields
\begin{align}
\mathrm{FPTE}(t)
\le \alpha^t\|\mathbf w_0-\widetilde{\mathbf w}_1\|
+\frac{2G}{\mu}\frac{A}{t},
\qquad \forall t\ge t_0,
\label{eq:FPTE_uniform_final}
\end{align}
where $A$ and $t_0$ are given in Proposition~\ref{prop:S_uniform_bound}.
As a result, it holds that $\lim_{t\to\infty}\mathrm{FPTE}(t)=0$. The bound on the TE for the uniform temporal weighting strategy is provided next.
\begin{theorem}
    \label{thm:TE_uniform}
    Under uniform weights \eqref{eq:uniform_weights}, the tracking error satisfies
\begin{align}
\mathrm{TE}(t)
\le \alpha^t\|\mathbf w_0-\widetilde{\mathbf w}_1\|
+\frac{2G}{\mu}\frac{A}{t} + \eta\kappa\Lambda G,
\quad \forall t\ge t_0,
\label{eq:TE_uniform_final}
\end{align}
where $A$ and $t_0$ are defined in Proposition~\ref{prop:S_uniform_bound}. Consequently, $\limsup_{t\to\infty}\mathrm{TE}(t)\le \eta\kappa\Lambda G$.
\end{theorem}
\noindent
The proof of Theorem~\ref{thm:TE_uniform} follows directly from \eqref{eq:bias_bound} and \eqref{eq:FPTE_uniform_final}.


\begin{remark}
Since the first term in \eqref{eq:TE_uniform_final} decays geometrically and the second term decays as $\mathcal O(1/t)$, the asymptotic error under uniform weighting is determined entirely by the bias term $\eta\kappa\Lambda G$. Therefore, under a constant step size, the fixed-point tracking contribution vanishes, while the residual TE floor is due to data heterogeneity across the network.
\end{remark}

\subsubsection{Exponentially discounted weights}
Another natural choice is to geometrically discount older samples and place more emphasis on recent observations. This corresponds to choosing weights of the form $a_i(t)\propto \gamma^{t-i}$ for all $i\le t$, where $\gamma \in (0,1)$ is the discount factor. Enforcing the constraint in \eqref{eq:weight_simplex} gives
\begin{align}
a_i(t)=\frac{1-\gamma}{1-\gamma^t}\gamma^{t-i},
\qquad i=1,\dots,t,
\label{eq:discounted_weights}
\end{align}
Under \eqref{eq:discounted_weights}, the temporally weighted objective satisfies
\begin{align}
\fbar_{t+1}(\mathbf w)
&= \sum_{i=1}^{t+1}\frac{1-\gamma}{1-\gamma^{t+1}}\gamma^{t+1-i}f_i(\mathbf w)\nonumber\\
&= \sum_{i=1}^{t}\frac{1-\gamma}{1-\gamma^{t+1}}\gamma^{t+1-i}f_i(\mathbf w)
+\frac{1-\gamma}{1-\gamma^{t+1}}f_{t+1}(\mathbf w)\nonumber\\
&= \gamma\frac{1-\gamma^t}{1-\gamma^{t+1}}\fbar_t(\mathbf w)
+\frac{1-\gamma}{1-\gamma^{t+1}}f_{t+1}(\mathbf w).
\label{eq:fbar_discounted_recursion}
\end{align}
Using \eqref{eq:fbar_discounted_recursion} in \eqref{eq:generic_fp_drift}, the fixed-point drift can be bounded as
\begin{align}
\|\widetilde{\mathbf w}_{t+1}-\widetilde{\mathbf w}_t\|
&\leq \frac{1}{\mu} \Big\|\frac{1-\gamma}{1-\gamma^{t+1}}
\big(
\nabla f_{t+1}(\widetilde{\mathbf w}_{t+1})
-\nabla \fbar_t(\widetilde{\mathbf w}_{t+1})
\big)\Big\|\nonumber\\
&\le \frac{2G}{\mu}\frac{1-\gamma}{1-\gamma^{t+1}},
\label{eq:discounted_fp_drift}
\end{align}
where the last step follows from the triangle inequality together with Lemma~\ref{lem:boundedness_compact}. Next, substituting \eqref{eq:discounted_fp_drift} into \eqref{eq:fpte_unrolled} yields
\begin{align}
\mathrm{FPTE}(t)
\le 
\alpha^t\|\mathbf w_0-\widetilde{\mathbf w}_1\|
{+}\frac{2G}{\mu}(1-\gamma)
\sum_{i=1}^{t-1}\frac{\alpha^{t-i}}{1-\gamma^{i+1}}.
\label{eq:FPTE_discounted_pre}
\end{align}
To bound the second term in \eqref{eq:FPTE_discounted_pre}, we use the following result.
\begin{prop}
\label{prop:S_discounted_bound}
Define $S_\gamma(t) \triangleq \sum_{i=1}^{t-1}  \frac{(1-\gamma)\alpha^{t-i}}{1 - \gamma^{i+1}}$ and let $A_\gamma\triangleq\max\bigl\{\frac{(1-\gamma^{t_0})S_\gamma(t_0)}{1-\gamma}, \frac{2\alpha}{1-\alpha}\bigr\}$, and $t_0 \triangleq \lceil\ln(\frac{1-\alpha}{1+\alpha-2\gamma\alpha})/\ln(\gamma)\rceil$. Then for all $t\geq t_0 $ we have  
\begin{align}
S_\gamma(t)\le A_\gamma\,\frac{1-\gamma}{1-\gamma^t}.
\label{eq:S_discounted_bound}
\end{align}
Moreover, $\lim_{t\to\infty} S_\gamma(t)=\frac{(1-\gamma)\alpha}{1-\alpha}$.
\end{prop}
\noindent
Proposition~\ref{prop:S_discounted_bound} is a direct restatement of the corresponding single-agent summation bound in \cite[Prop.~2]{TV_asilomar25}
Combining \eqref{eq:FPTE_discounted_pre} with Proposition~\ref{prop:S_discounted_bound}, we obtain
\begin{align}
\mathrm{FPTE}(t)
\le
\alpha^t\|\mathbf w_0-\widetilde{\mathbf w}_1\|
+\frac{2G}{\mu}A_\gamma\frac{1-\gamma}{1-\gamma^t},
\; \forall t\ge t_0,
\label{eq:FPTE_discounted_final}
\end{align}
where $A_\gamma$ and $t_0$ are given in Proposition~\ref{prop:S_discounted_bound}. From \eqref{eq:FPTE_discounted_final}, it further holds that $\limsup_{t\to\infty}\mathrm{FPTE}(t)
\le
\frac{2G}{\mu}\frac{(1-\gamma)\alpha}{1-\alpha}$.
TE guarantees for the discounted weights are provided next.

\begin{theorem}
\label{thm:TE_discounted}
Under discounted weights \eqref{eq:discounted_weights}, the tracking error satisfies
\begin{align}
\mathrm{TE}(t)
\le
\alpha^t\|\mathbf w_0-\widetilde{\mathbf w}_1\|
+\frac{2G}{\mu}A_\gamma\frac{1-\gamma}{1-\gamma^t}
+\eta\kappa\Lambda G, 
\label{eq:TE_discounted_final}
\end{align}
for all $t\ge t_0$, where $A_\gamma$ and $t_0$ are defined in Proposition~\ref{prop:S_discounted_bound}. Consequently,
\begin{align}
\limsup_{t\to\infty}\mathrm{TE}(t)
\le
\frac{2G}{\mu}\frac{(1-\gamma)\alpha}{1-\alpha}
+\eta\kappa\Lambda G.
\label{eq:ATE_discounted}
\end{align}
\end{theorem}
\noindent
The proof of Theorem~\ref{thm:TE_discounted} follows directly from \eqref{eq:bias_bound}, \eqref{eq:FPTE_discounted_final}, and Proposition~\ref{prop:S_discounted_bound}.
\begin{remark}
In contrast to uniform weighting, discounted weighting induces a non-vanishing fixed-point tracking floor. This occurs because the newest samples across the agents always induce a fixed-point drift of order $1-\gamma$, so the algorithm never fully catches up to the moving fixed point. 
\end{remark}

\section{Numerical Results}
In this section, we present numerical experiments to illustrate the tracking behavior predicted by the analysis. Although the analysis focuses on the scalar case, we simulate its natural vector extension. To capture decentralized streaming behavior, we consider a quadratic loss at each agent and time index. Specifically, for agent $n \in \mathcal{V}$ and time $t\ge 1$, the instantaneous loss is given by
\begin{align}
\ell_{n,t}(\mathbf{w})
=\frac{1}{2}\bigl(\mathbf{w}-\mathbf{c}_{n,t}\bigr)^\top \mat{A}_{n,t}\bigl(\mathbf{w}-\mathbf{c}_{n,t}\bigr),
\label{eq:num_local_quadratic_vec}
\end{align}
for all $\mathbf{w}\in\mathbb{R}^d$, where $\mathbf{c}_{n,t}\in\mathbb{R}^d$ captures the streaming-data process and $\mat{A}_{n,t}\in\mathbb{R}^{d\times d}$ is a positive definite Hessian matrix. In our experiments, we generate $\mat{A}_{n,t}=\mathrm{diag}(\lambda_{n,t}^1,\ldots,\lambda_{n,t}^d)$, where $\lambda_{n,t}^j\sim\mathrm{Unif}[\mu,L]$ independently across $n$, $t$, and $j$. Thus, $\ell_{n,t}$ in \eqref{eq:num_local_quadratic_vec} is $\mu$-strongly convex and $L$-smooth. It is easy to verify that the instantaneous network objective $F_t(\mathbf{w})=\frac{1}{N}\sum_{n=1}^N \ell_{n,t}(\mathbf{w})$ is minimized at $\mathbf{w}_t^*
=
\big(\sum_{n=1}^N \mat{A}_{n,t}\big)^{-1}
\big(\sum_{n=1}^N \mat{A}_{n,t}\mathbf{c}_{n,t}\big)$.
Accordingly, the time-weighted objective
$\Fbar_t(\mathbf{w})=\sum_{i=1}^t a_i(t)F_i(\mathbf{w})$
admits the minimizer $\overline{\mathbf{w}}_t^*
=
\big(\sum_{i=1}^t a_i(t)\sum_{n=1}^N \mat{A}_{n,i}\big)^{-1}
\big(\sum_{i=1}^t a_i(t)\sum_{n=1}^N \mat{A}_{n,i}\mathbf{c}_{n,i}\big)$,
enabling closed-form computation of the TE in \eqref{eq:TE_def}.

The temporal variation is driven by $\mathbf{c}_{n,t}$, which evolves according to a bounded Gaussian random walk. For each coordinate $j =1,\cdots,d$, it has the form:
\begin{align*}
[\mathbf{c}_{n,t+1}]_j =
\max\!\left(
-C_{\max},
\min\bigl([\mathbf{c}_{n,t}]_j+[\mathbf{z}_{n,t+1}]_j,\; C_{\max}\bigr)
\right),
\end{align*}
where $\mathbf{z}_{n,t}\sim \mathcal{N}(\mathbf{0},\sigma^2 \mat{I})$ are i.i.d. across agents and time, and $\mathbf{c}_{n,0}\sim\mathrm{Unif}[-C_{\max},C_{\max}]^d$. We set $\mathbf{w}_{n,0}=\mathbf{0}$ for all $n$.

For the communication network, we place $N = 30$ agents uniformly at random in a circular disk and generate an undirected random geometric graph. Two agents are connected if their Euclidean distance is at most a threshold $r_{\mathrm{th}}$. We increase $r_{\mathrm{th}}$ when needed to ensure connectivity, and then construct a symmetric doubly stochastic mixing matrix $\mat{M}$ via the Metropolis rule:
\begin{align}
[\mat{M}]_{ij}=
\begin{cases}
\displaystyle \frac{1}{1+\max\{d_i,d_j\}}, & (i,j)\in\mathcal{E},\ i\neq j,\\[1ex]
\displaystyle 1-\sum_{j\neq i}[\mat{M}]_{ij}, & i=j,\\
0, & \text{otherwise},
\end{cases}
\label{eq:num_metropolis}
\end{align}
where $d_i$ denotes the degree (number of neighbors in the network graph) of agent $i$. We choose the model dimension $d=50$. The streaming process parameters are chosen as $C_{\max}=10, \sigma^2=1,\mu=0.01$ and $L=0.1$.
We evaluate the tracking error using the stacked iterate $\mathbf{w}_t \triangleq [\mathbf{w}_{1,t}^\top,\ldots,\mathbf{w}_{N,t}^\top]^\top \in \mathbb{R}^{Nd}$, and accordingly, the TE is then computed as
$\mathrm{TE}(t)
\triangleq
\|\mathbf{w}_t-(\mathbf{1}\otimes \overline{\mathbf{w}}_t^*)\|_2$.
All results are reported in terms of root-mean-squared (RMS) TE, $
\sqrt{\frac{1}{R}\sum_{r=1}^R \mathrm{TE}_r(t)^2}$,
averaged over $R=50$ Monte Carlo runs. We consider both uniform temporal weights \eqref{eq:uniform_weights} and exponentially discounted weights \eqref{eq:discounted_weights}, with discount factor $\gamma=0.7$. Figures~\ref{fig:te_uniform_multiE}--\ref{fig:te_discounted_multiE} illustrate the corresponding tracking behavior for several choices of the number of inner DGD iterations $E$, under step size $\eta=0.05$.


\begin{figure*}[t]
    \centering
    \subfloat[Uniform temporal weights.]{
        \includegraphics[width=0.44\textwidth]{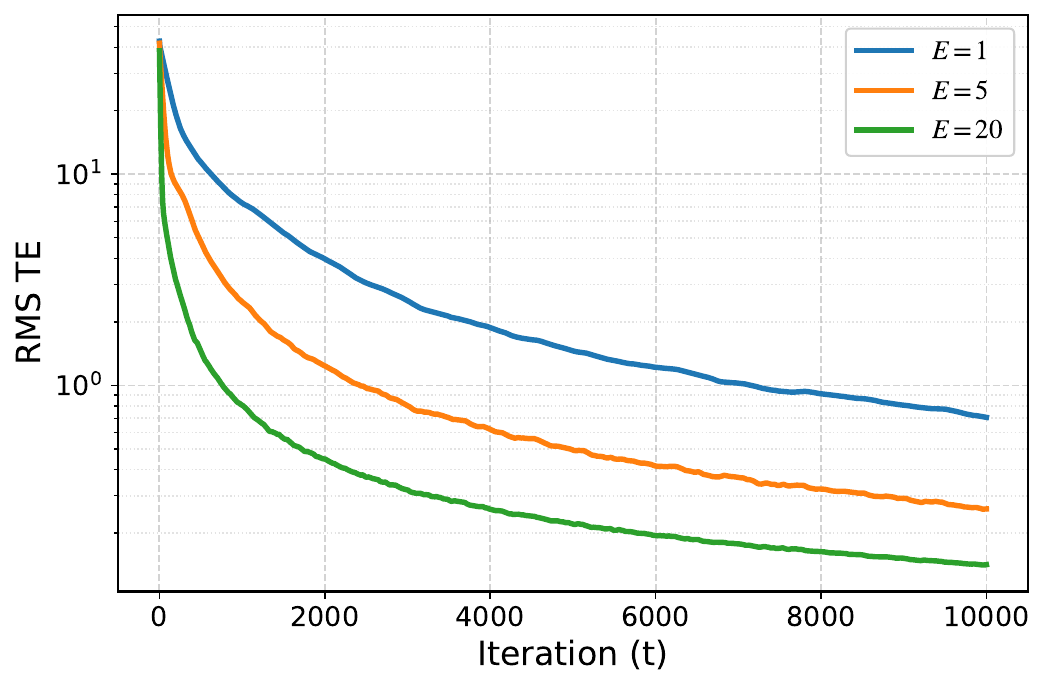}
        \label{fig:te_uniform_multiE}
    }
    \hfill
    \subfloat[Exponentially discounted temporal weights with $\gamma=0.7$.]{
        \includegraphics[width=0.44\textwidth]{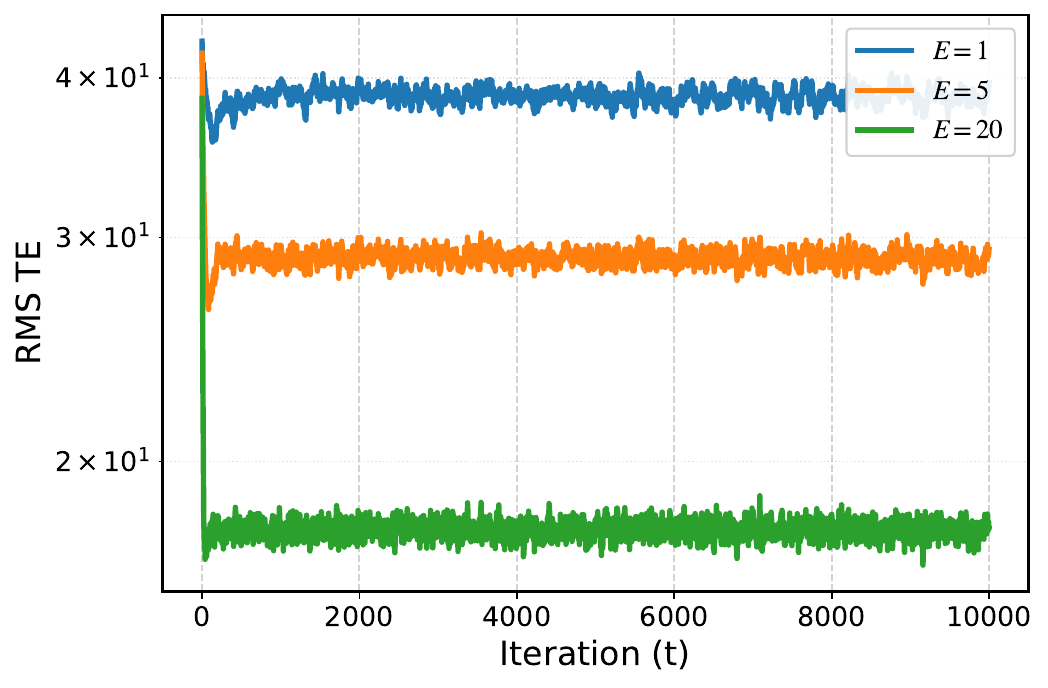}
        \label{fig:te_discounted_multiE}
    }
    \caption{RMS tracking error versus iteration index for different numbers of inner DGD steps $E$, with $\eta=0.05$. 
    }
    \label{fig:te_both}
    \vspace{-2mm}
\end{figure*}
Figure~\ref{fig:te_uniform_multiE} shows that, under uniform temporal weighting, increasing the number of DGD iterations per time index improves the transient tracking behavior by reducing the fixed-point tracking component. As predicted by the theory, the influence of the newest sample decreases with time, which leads to a decaying fixed-point drift, and hence a vanishing FPTE. Consequently, aligned with Theorem~\ref{thm:TE_uniform}, the RMS TE eventually approaches a regime dominated by the bias term, yielding a non-zero floor determined by the step size, network topology, and data heterogeneity.

Figure~\ref{fig:te_discounted_multiE} plots the RMS TE under discounted weighting with $\gamma=0.7$. In this case, for all the considered choices of $E$, a larger non-zero TE floor is attained. This is consistent with   Theorem~\ref{thm:TE_discounted}: exponentially discounted weighting induces a non-vanishing FPTE. While increasing $E$ improves tracking and lowers the final TE level, unlike uniform weighting, the FPTE does not vanish because old samples are continually forgotten and the effective objective keeps evolving at a non-vanishing rate. 

\section{Conclusion}
In this work, we studied decentralized optimization with streaming data under a temporal-weighting formulation. Focusing on DGD, we derived TE guarantees through a fixed-point analysis that decomposes the tracking error into a fixed-point tracking term and a bias term. Specializing the analysis to two natural weighting strategies, we showed that uniform weighting yields a vanishing fixed-point tracking contribution of order $\mathcal{O}(1/t)$, whereas the discounted weighting induces a non-vanishing steady-state contribution governed by the discount factor. Overall, the resulting weight-specific bounds provide interpretable characterizations of learning from streaming data and clarify how the temporal weighting rule and the per-step DGD budget affect the asymptotic tracking performance. Numerical experiments corroborate the theoretical findings.

\appendix
\label{app:A}
\begin{proof}[Proof of Lemma~\ref{lem:generic-drift}]
Fix any $t\ge 1$. From the definition of $\phi_t$ in \eqref{eq:dgd_map} and the fixed-point relation \eqref{eq:fixed_point_def}, we have
\begin{align}
    (\mat{I}-\mat{M})\widetilde{\mathbf w}_t+\eta \nabla\overline f_t(\widetilde{\mathbf w}_t)=\mathbf{0}.
    \label{eq:fixed-point-rel-unified}
\end{align}
Let $\mathbf{d}_t \triangleq \widetilde{\mathbf w}_{t+1}-\widetilde{\mathbf w}_t$ denote the fixed-point drift at time $t$. Writing \eqref{eq:fixed-point-rel-unified} at times $t$ and $t+1$ and subtracting, we obtain:
\begin{align}
(\mat{I}-\mat{M})\mathbf{d}_t
+\eta\bigl(\nabla\overline f_{t+1}(\widetilde{\mathbf w}_{t+1})
-\nabla\overline f_t(\widetilde{\mathbf w}_t)\bigr)=\mathbf{0}.
\label{eq:dgd_drift_raw}
\end{align} 
Next, by the mean-Hessian theorem \cite{MHT}, applied to $\fbar_t(\cdot)$ (satisfying Assumption~\ref{ass:smooth_sc}) at the points $\widetilde{\mathbf w}_t$ and $\widetilde{\mathbf w}_{t+1}$, there exists a symmetric matrix $\mat{A}$ such that $\mu \mat{I}\preceq \mat{A}\preceq L \mat{I}$ and
\begin{align}
\nabla\overline f_t(\widetilde{\mathbf w}_t)-\nabla\overline f_t(\widetilde{\mathbf w}_{t+1})
= \mat{A}(\widetilde{\mathbf w}_t-\widetilde{\mathbf w}_{t+1}).
\label{eq:MHT_drift_rel_app}
\end{align}
Substituting \eqref{eq:MHT_drift_rel_app} into \eqref{eq:dgd_drift_raw} gives
\begin{align*}
    \bigl[(\mat{I}-\mat{M})+\eta \mat{A}\bigr]\mathbf{d}_t
=-\eta\bigl(\nabla\overline f_{t+1}(\widetilde{\mathbf w}_{t+1})
-\nabla\overline f_t(\widetilde{\mathbf w}_{t+1})\bigr).
\end{align*}
Since $\mat{I}-\mat{M}\succeq 0$ and $\mat{A}\succeq \mu \mat{I}$, we have
$\lambda_{\min}((\mat{I}-\mat{M})+\eta \mat{A})\ge \eta\mu$. Therefore, $\|\mathbf{d}_t \|
\le \frac{\eta}{\eta\mu}\,
\bigl\|\nabla\overline f_{t+1}(\widetilde{\mathbf w}_{t+1})
-\nabla\overline f_t(\widetilde{\mathbf w}_{t+1})\bigr\|$, which is exactly \eqref{eq:generic_fp_drift}.
\end{proof}

\begin{proof}[Proof of Lemma~\ref{lem:boundedness_compact}]
We first bound the global minimizer $\overline w_t^*$. Since $\Fbar_t(\cdot)$ is $\mu$-strongly convex and minimized at $\overline w_t^*$, evaluating the strong-convexity inequality at $w=0$ yields
\begin{align}
    \bigl|\overline{w}_t^*\bigr|^2
    \leq \frac{2}{\mu}
         \bigl(\overline{F}_t(0) - \overline{F}_t(\overline{w}_t^*)\bigr).
         \label{eq:global-min-bnd_interm}
\end{align}
Using the definition of $\Fbar_t(\cdot)$ in \eqref{eq:main_tv_prob}, we obtain
\begin{align*}
\Fbar_t(0)-\Fbar_t(\overline w_t^*)
&=
\sum_{i=1}^t a_i(t)\frac{1}{N}\sum_{n=1}^N
\bigl(\ell_{n,i}(0)-\ell_{n,i}(\overline w_t^*)\bigr)
\\
&\le
\sum_{i=1}^t a_i(t)\frac{1}{N}\sum_{n=1}^N
\bigl(\ell_{n,i}(0)-\ell_{n,i}(w_{n,i}^*)\bigr),
\end{align*}
since $w_{n,i}^*$ minimizes $\ell_{n,i}(\cdot)$. By $L$-smoothness and the optimality of $w_{n,i}^*$, it holds that $\ell_{n,i}(0)-\ell_{n,i}(w_{n,i}^*)
\le \frac{L}{2}|w_{n,i}^*|^2
\le \frac{L}{2}C^2$, where the last step uses Assumption~\ref{ass:bounded_minimizers}. Since $\sum_{i=1}^t a_i(t)=1$, it follows that
\begin{align}
\Fbar_t(0)-\Fbar_t(\overline w_t^*)
\le \frac{L}{2}C^2.
\label{eq:global-min-bnd_interm1}
\end{align}
Combining \eqref{eq:global-min-bnd_interm} and \eqref{eq:global-min-bnd_interm1} yields the bound in \eqref{eq:minimizer_fixedpoint-bnd}. Next, we bound the fixed point $\widetilde{\mathbf w}_t$. The fixed-point relation reads
\begin{align}
(\mat{I}-\mat{M})\widetilde{\mathbf w}_t+\eta\nabla\overline f_t(\widetilde{\mathbf w}_t)=\mathbf{0}.
\label{eq:fp_dgd}
\end{align}
From the $\mu$-strong convexity of $\overline{f}_t(\cdot)$ (Assumption \ref{ass:smooth_sc}), we have $\overline{f}_t(\mathbf{0})
    \geq   \overline{f}_t(\widetilde{\mathbf{w}}_t) - \nabla \overline{f}_t(\widetilde{\mathbf{w}}_t)^\top\widetilde{\mathbf{w}}_t + \frac{\mu}{2}\bigl\|\widetilde{\mathbf{w}}_t\bigr\|^2$. From \eqref{eq:fp_dgd},
$\nabla\overline f_t(\widetilde{\mathbf w}_t)=-(1/\eta)(\mat{I}-\mat{M})\widetilde{\mathbf w}_t$ which further yields
$\overline f_t(\mathbf 0)
\ge \overline f_t(\widetilde{\mathbf w}_t)
+\frac{1}{\eta}\widetilde{\mathbf w}_t^\top(\mat{I}-\mat{M})\widetilde{\mathbf w}_t
+\frac{\mu}{2}\|\widetilde{\mathbf w}_t\|^2$. 
Next, since $\mat{M}$ is symmetric with $\lambda_{\max}(\mat{M})=1$, it follows that $\mat{I}-\mat{M}\succeq 0$. Therefore, the term $\frac{1}{\eta}\widetilde{\mathbf w}_t^\top(\mat{I}-\mat{M})\widetilde{\mathbf w}_t$ is nonnegative and can be dropped, giving
\begin{align}
\|\widetilde{\mathbf w}_t\|^2
\le \frac{2}{\mu}\bigl(\overline f_t(\mathbf 0)-\overline f_t(\widetilde{\mathbf w}_t)\bigr).
\label{eq:fixedpoint-bnd_interm}
\end{align}
Moreover, from the definition of $\overline{f}_t(\cdot)$, we have
\begin{align}
\fbar_t(\mathbf 0)-\fbar_t(\widetilde{\mathbf w}_t)
&=
\sum_{i=1}^t a_i(t)\sum_{n=1}^N
\bigl(\ell_{n,i}(0)-\ell_{n,i}([\widetilde{\mathbf w}_t]_n)\bigr)\nonumber\\
&\le
\sum_{i=1}^t a_i(t)\sum_{n=1}^N
\bigl(\ell_{n,i}(0)-\ell_{n,i}(w_{n,i}^*)\bigr)\nonumber\\
&\le
\sum_{i=1}^t a_i(t)\sum_{n=1}^N
\frac{L}{2}|w_{n,i}^*|^2
\le \frac{NL}{2}C^2.
\label{eq:fixedpoint-interm1}
\end{align}
Substituting \eqref{eq:fixedpoint-interm1} into \eqref{eq:fixedpoint-bnd_interm} proves the desired result in \eqref{eq:minimizer_fixedpoint-bnd}.

We next prove the gradient bounds. Let
$\overline{\mathbf m}_t^*=\arg\min_{\mathbf w\in\R^N}\fbar_t(\mathbf w)$.
Since $\fbar_t(\mathbf w)=\sum_{n=1}^N\sum_{i=1}^t a_i(t)\ell_{n,i}(w_n)$ is separable across coordinates, each component $[\overline{\mathbf m}_t^*]_n$ minimizes the scalar function $\sum_{i=1}^t a_i(t)\ell_{n,i}(w)$. By the same argument used above for $\overline w_t^*$ in \eqref{eq:global-min-bnd_interm}--\eqref{eq:global-min-bnd_interm1}, it follows that $|[\overline{\mathbf m}_t^*]_n|\le C\sqrt{\kappa}$, and hence $
\|\overline{\mathbf m}_t^*\|\le C\sqrt{N\kappa}$. Next, using $L$-smoothness of $\fbar_t(\cdot)$ and the optimality condition
$\nabla \fbar_t(\overline{\mathbf m}_t^*)=\mathbf 0$, for any $i,t\ge 1$ we have
\begin{align*}
\|\nabla\overline f_t(\widetilde{\mathbf w}_i)\|
&= \|\nabla\overline f_t(\widetilde{\mathbf w}_i)-\nabla\overline f_t(\overline{\mathbf m}_t^*)\|\\
&\leq L\|\widetilde{\mathbf w}_i-\overline{\mathbf m}_t^*\|
\leq  L\bigl(\|\widetilde{\mathbf w}_i\|+\|\overline{\mathbf m}_t^*\|\bigr) \leq 2LC\sqrt{N\kappa}.
\end{align*}
This proves the first bound in \eqref{eq:grad-bnd-G}. The bound on $\|\nabla f_t(\widetilde{\mathbf w}_i)\|$ follows analogously. 
Finally, since $\overline{\mathbf w}_t^*=\overline w_t^*\mathbf 1$, from \eqref{eq:minimizer_fixedpoint-bnd}, we have
$\|\overline{\mathbf w}_t^*\| \le C\sqrt{N\kappa}$.
Using again the $L$-smoothness of $\fbar_t(\cdot)$ and
$\nabla \fbar_t(\overline{\mathbf m}_t^*){=}\mathbf 0$, we obtain $\|\nabla \overline f_t(\overline{\mathbf w}_t^*)\| {=} \|\nabla \overline f_t(\overline{\mathbf w}_t^*)-\nabla \overline f_t(\overline{\mathbf m}_t^*)\| \leq L\|\overline{\mathbf w}_t^*-\overline{\mathbf m}_t^*\|
{\le} L\bigl(\|\overline{\mathbf w}_t^*\|+\|\overline{\mathbf m}_t^*\|\bigr)
{\le} 2LC\sqrt{N\kappa}$, 
which proves the desired bound.

\end{proof}

\balance
\bibliographystyle{IEEEtran}
\bibliography{Refs}

\end{document}